\documentclass[13pt,a4paper]{article}

\usepackage[ruled]{algorithm2e}
\usepackage{algpseudocode}
\usepackage{amsthm} 
\usepackage{amsmath}
\usepackage{amssymb}
\usepackage{mathtools}

\DeclarePairedDelimiter\floor{\lfloor}{\rfloor}

\newcommand{\BigO}[1]{\ensuremath{\operatorname{O}\bigl(#1\bigr)}}

\newtheorem{theorem}{Theorem}[section]

\newtheorem{lemma}[theorem]{Lemma}

\newtheorem{definition}[theorem]{Definition}
\newtheorem{notation}[theorem]{Notation}
\newtheorem{remark}[theorem]{Remark}

\begin{document}

%\begin{frontmatter}

\title{The Takeoff Towards Optimal Sorting Networks}
\author{Martin Marinov\\ David Gregg}
\maketitle

%\title{ Towards Optimal Sorting Networks: The Third Level }

%\author{Martin Marinov\fnref{label:Martin}}
%\ead{marinovm@tcd.ie}
%\fntext[label:Martin]{Corresponding author. Dept. of Computer Science, Trinity College Dublin, Ireland. Work supported by the Irish Research Council (IRC).}

%\author{David Gregg\fnref{label:David}}
%\ead{dgregg@cs.tcd.ie}
%\fntext[label:David]{Lero, Trinity College Dublin.}

\begin{abstract}

A complete set of filters $F_n$ for the optimal-depth $n$-input sorting network problem is such that if there exists an $n$-input sorting network of depth $d$ then there exists one of the form $C \oplus C'$ for some $C \in F_n$. Previous work on the topic presents a method for finding complete set of filters $R_{n, 1}$ and $R_{n, 2}$ that consists only of networks of depths one and two respectively, whose outputs are minimal and representative up to permutation and reflection. Our main contribution is a practical approach for finding a complete set of filters $R_{n, 3}$ containing only networks of depth three whose outputs are minimal and representative up to permutation and reflection. In previous work, we have developed a highly efficient algorithm for finding extremal sets ( i.e. outputs of comparator networks; itemsets; ) up to permutation. In this paper we present a modification to this algorithm that identifies the representative itemsets up to permutation and reflection. Hence, the presented practical approach is the successful combination of known theory and practice that we apply to the domain of sorting networks. For $n < 17$, we empirically compute the complete set of filters $R_{n, 2}$, $R_{n, 3}$, $R_{n, 2} \upharpoonright w $ and $R_{n, 3}^w$ of the representative minimal up to permutation and reflection $n$-input networks, where all but $R_{n, 2}$ are novel to this work.

\end{abstract}

%\begin{keyword}
%sorting network \sep minimal depth \sep third level\sep computer assisted proof \sep extremal sets \sep extremal sets up to permutation
%\end{keyword}
%
%\end{frontmatter}

\section{Introduction}

%The problem addressed in this paper is the \emph{Bose-Nelson sorting problem~\cite{Bose:1962:SP}} of finding the most efficient sorting network.

%\noindent
%One practical use of sorting networks is in software applications such as decryption algorithms where we need a fast method of sorting numbers with a fixed sequence of comparators. The use of sorting networks in such algorithms prevents ``timing attacks'' --- reverse engineering using precise time measurements. Another use of sorting networks is to implement a sorting/comparator network in hardware.

%\noindent

%\noindent
A sorting network is an abstract mathematical model designed to sort numbers in a predetermined sequence of comparators. A sorting network consists of $n$ wires and comparators between pairs of wires such that any input of $n$ numbers is sorted by the network, where one wire corresponds to one number. The two most common measures of sorting networks are the total number of comparators --- \emph{Bose-Nelson's sorting problem~\cite{Bose:1962:SP}} --- and the number of network levels, also referred to as depth. In this paper, we extend the knowledge on the problem of empirically searching for sorting networks of minimal depth by significantly reducing the number of candidate network levels of depth three that need to be considered by any algorithm. For example for $n = 11$, $R_{11, 3}$ contains about $10$ thousand comparator networks in comparison to existing methods that would consider roughly $1.7$ million eleven-input comparator networks of depth three, that is about $169$ times fewer networks to consider in comparison to applying all levels.

%\clearpage
\section{Related Work and Contributions}

%\noindent
Knuth \cite{Knuth73} showed the optimal depth sorting networks for all $n \leq 8$. He also presents the zero-one principle of sorting networks which states that if a comparator network sorts all $2^n$ binary strings of length $n$ then it is a sorting network.

%\noindent
Parberry~\cite{Parberry89} presented a computer assisted proof for the minimal depth of a nine and ten-input sorting networks. He significantly reduced network level candidates for the first two levels, in comparison to the naive approach, by exploiting symmetries of the networks (referred to as first and second normal form \cite{Parberry89}).

Bundala~\cite{BundalaCCSZ14_Optimal_Depth} presented a computer assisted proof for the optimal depths of networks with eleven to sixteen (inclusive) inputs. He also managed to significantly reduce the number of candidates for the second layer in comparison to Parberry's approach, by considering only networks whose outputs are minimal representative up to permutation and reflection. Similar work for the second level is also presented by Michael Codish in \cite{CodishCS14:Two_Layer_Prefix}. Bundala's algorithm for finding sorting networks of optimal depth is based on a SAT encoding of the optimal depth sorting networks problem, which uses the set of candidate two-layer networks as a fixed entry point. Some extra pruning techniques are presented and they use a state of the art ``off the shelf'' SAT solver to find the optimal depth sorting networks for all $n \leq 16$. The work presented in this paper can be used to speedup Bundala's algorithm substantially, as described in our experiments section~\ref{sec:experiments}.

We~\cite{Marinov:ExtremalSets:Permutation} presented a highly efficient practical algorithm for finding the minimal representative itemsets over a domain $D$ up to a permutation of $D$. This algorithm can be applied to reduce the number of candidates for the second layer as described in Bundala and Codish, although \cite{BundalaCCSZ14_Optimal_Depth} and \cite{CodishCS14:Two_Layer_Prefix} present an extra pruning method using reflection.

\subsection{ Problem Statement}

The problem addressed in this paper is that of reducing the candidate networks of depths two and three that need to be considered when searching for sorting networks of optimal depth.

\subsection{Motivation}

The importance of finding the $R_{n,3}$ and $R_{n,3} \upharpoonright w$ set of filters is easily seen from Bundala's algorithm for finding sorting networks of optimal depth and is explicitly stated in the future work section of \cite{BundalaCCSZ14_Optimal_Depth}. Bundala's algorithm can be easily adapted to use prefixes of exactly three layers as an entry point to the SAT encoding of the problem, given that the presented algorithm \cite{BundalaCCSZ14_Optimal_Depth} uses exactly two layers as the entry point. Hence, such a reduction of the search space would result in a faster such SAT-solver-based algorithm for finding sorting networks of optimal depth.

It is important to mention that the techniques described by Bundala~\cite{BundalaCCSZ14_Optimal_Depth} and Codish~\cite{CodishCS14:Two_Layer_Prefix} for reducing the number of candidate networks of depth two cannot be easily extended for networks of depth three because they are based on finite case studies/identification around the second level. In other words, they provide a regular expression for generating the networks of depth two whose outputs are minimal and representative up to permutation and reflection. Identifying all of the cases to derive a regular expression for the third level is a possible solution to the problem but could prove an immensity hard one.

Moreover, using our approach we manage to further reduce the required set of two-layered networks that needs to be considered when finding optimal depth sorting networks. Since Bundala showed that is a sufficient to find a subset of all of the $2^n$ possible inputs for which no $n$-input sorting network of depth $d$ exists that sorts all of them when proving depth optimality. Intuitively, the more we reduce the initial input set, the lower the number of minimal representative up to permutation and reflection networks of depth two is enough to be considered. We present experiments showing the achieved reduction for all $n$-input two layered networks for $n < 17$.

%\clearpage
\subsection{Contributions}

\begin{itemize}
	\item \emph{Modified algorithm for finding minimal itemsets up to permutation and reflection} --- we took an existing algorithm for finding minimal itemsets up to permutation which a dataset (a collection of itemsets). We present a modification which is linear (in terms of the number of itemsets) in time and space to find the ones which are minimal up to permutation and reflection.
	
	\item \emph{Empirically find $R_{n, 2} \upharpoonright w $ for all $n < 17$} --- this is a direct improvement of Bundala's technique of not considering all of the $2^n$ inputs to determine that no $n$-input sorting network of depth $d$ exists. We take one step further to find the minimal up to permutation and reflection itemsets after applying the input set reduction (described in Experiment~3~\cite{BundalaCCSZ14_Optimal_Depth}).
	
	\item \emph{Empirically find $R_{n, 3}$ and $R_{n, 3}^w$ for all $n < 14$} --- we experimentally evaluated the modified algorithm to find the three layered $n$-input comparator networks whose outputs (itemsets) are minimal up to permutation and reflection. The set $R_{n, 3}$ is generated by applying all network levels to $R_{n, 2}$ and then finding the minimal up to permutation and reflection ones, whereas $R_{n, 3}^w$ is derived by applying all levels to all itemsets in the set $R_{n, 2} \upharpoonright w $  and then reducing.
	
\end{itemize}

%\clearpage
\section{Background on Sorting Networks}

\subsection{ Formal Definition of Comparator and Sorting Networks }

\begin{definition}
\label{ComparatorNetworkDefinition}
A \emph{generalized comparator} is an ordered pair $\langle i, j \rangle$ such that $1 \leq i \neq j \leq n$. A generalized comparator is a \emph{comparator} or \emph{min-max comparator} if $i < j$. The values $i$ and $j$ are referred to as \emph{channels}.
A \emph{generalized level} $L$ is a set of generalized comparators such that each channel is involved in at most one generalized comparator, formally if $\langle a, b \rangle, \langle c, d \rangle \in L$ then $\lvert \{ a, b, c, d \} \lvert = 4 $. A generalized level is a \emph{level} or \emph{min-max level} if it consists only of (min-max) comparators. The set of all (min-max) levels is denoted as $G_n$, as described by Bundala~\cite{BundalaCCSZ14_Optimal_Depth}.
A \emph{generalized $n$-input comparator network} is a vector $\langle L_1, L_2, \dots, L_d, n \rangle $, where $L_1, L_2, \dots, L_d$ are generalized levels, and $n$ is a positive integer. A generalized $n$-input comparator network is called an \emph{$n$-input comparator network} if it consists only of (min-max) levels. Let $C = \langle L_1, L_2, \dots, L_d, n \rangle$ be an $n$-input comparator network, we define the size of $C$ as the number of levels, i.e. $|C| = d$.
\end{definition}

%\noindent
So far we have formally defined the structure of a (generalized) comparator network. We need to define the \emph{output} of applying a comparator network to an \emph{input}, where an input is an n-bit binary string \cite{Knuth73}. Applying a network to an input permutes the input vector. Hence, for any fixed input we can define a permutation that models the network behaviour when applied to that particular input.

\begin{notation}
\label{PermutationDefinition}
Denote the set of all permutations of n elements as $\Pi_n = \{ \pi : \{1, 2, \dots, n\} \longmapsto \{1, 2, \dots, n\}$ $\lvert$ $\pi$ is bijective $\}$. Let $v = \langle a_1, a_2, \dots \rangle$ be a vector. Denote by $v_i$ the $i$-th coordinate of $v$, namely $v_i = a_i$.
\end{notation}

\begin{definition}
\label{def:ComparatorNetworkEvaluation}
An \emph{input} is a vector $\mathbf{x} \in \{0,1\}^n$ as per Knuth's~\cite{Knuth73} zero-one principle. Denote by $I_n$ the set of all inputs. The evaluation of a generalized $n$-input comparator network $C = \langle L_1, \dots, L_d, n \rangle $ in channel $i$ at level $k$ on input $x$ is the two dimensional vector $e_x(i, k)$ where:

\[
 e_x(i, k) =
  \begin{cases} 
	\langle x_i, i \rangle		& \text{ $ k = 0 $}											\\
    e_x(i, k-1)		& \text{ $ \langle i, j \rangle \in L_k $ and $ e_x(i, k-1)_1 <= e_x(j, k-1)_1 $}	\\
    e_x(j, k-1)		& \text{ $ \langle i, j \rangle \in L_k $ and $ e_x(i, k-1)_1 >  e_x(j, k-1)_1  $}	\\
    e_x(i, k-1)		& \text{ $ \langle j, i \rangle \in L_k $ and $ e_x(i, k-1)_1 >= e_x(j, k-1)_1 $}	\\
    e_x(j, k-1)		& \text{ $ \langle j, i \rangle \in L_k $ and $ e_x(i, k-1)_1 <  e_x(j, k-1)_1$}	\\
	e_x(i, k-1)		& \text{ otherwise }										\\
  \end{cases}
\]

%\noindent
The output of applying $C$ to $x$ is \emph{$V_C(x) = \langle e_x(1, d)_1, \dots, e_x(n, d)_1 \rangle \in I_n$}. The permutation of the coordinates when applying $C$ to  $x$ is \emph{$P_C(x) = \langle e_x(1, d)_2, \dots , e_x(n, d)_2 \rangle \in \Pi_n $}.
\end{definition}

%\noindent
Intuitively, we say that a vector in $I_n$ is sorted if its values are non-decreasing left-to-right, and a sorting network is one which sorts all possible $2^n$ input vectors. More formally:

\begin{definition}
\label{SortingNetworkDefinition}
The vector $\langle x_1, x_2, \dots, x_n \rangle \in \{0,1\}^n$ is \emph{sorted} iff $x_i <= x_{i + 1}$ for all $1 \leq i < n$.
A \emph{generalized sorting network} is a generalized $n$-input comparator network for which there exists a permutation $\pi \in \Pi_n$ such that $\pi( V_C(x) )$ is sorted for all inputs $ x \in I_n$.
A \emph{sorting network} is an $n$-input comparator network such that $ V_C( x ) $ is sorted for all inputs $ x \in I_n$.
\end{definition}

\begin{theorem}
\label{FloydKnuthGeneralizedTheorem}
For every generalized sorting network there is a sorting network with the same size and depth. If the former has only min-max comparators in the first k levels, then the latter is identical in the first k levels.
\end{theorem}

\begin{proof}
See Knuth~\cite{Knuth73}.
\end{proof}

\subsection{Known Overlapping Theory}
\label{sec:background:theory:overlap}

The following definitions, lemmas, theorems and proofs presented in this section~\ref{sec:background:theory:overlap} are developed independently of the work by Michael Codish \cite{CodishCS14a_The_End_Game} \cite{CodishCS14:Two_Layer_Prefix} and Daniel Bundala \cite{BundalaCCSZ14_Optimal_Depth}. But since the results described in this section have already been published, we label them as known properties of sorting networks.

\begin{definition}
\label{OutputSetDefinition}
Let the \emph{output set} of a comparator network $C$ be $S_C = \{ V_C(x) \lvert x \in I_n \}$. Let the set of all already sorted inputs $T_n = \{ (x_1, x_2, \dots, x_n)$ $\lvert$  $x_{i < j} = 0, x_{i >= j} = 1$ for $ 1 \leq j \leq n + 1 \}$.
\end{definition}

\begin{definition}
\label{ComparatorNetworkUnionLevelDefinition}
Let $A$ and $B$ be $n$-input comparator networks, where $A = \langle A_1, A_2, \dots, A_d, n \rangle $, $B = \langle B_1, B_2, \dots, B_k, n \rangle $, and let $L$ be a level. Define the concatenations $A \oplus L = \langle A_1, \dots, A_d, L, n \rangle $ and $A \oplus B = A \oplus B_1 \oplus B_2 \oplus \dots \oplus B_k$. Note that $\oplus$ is associative.
\end{definition}

\begin{theorem}
\label{MinimalOutputSetTheorem}
Let $A$, $B$ and $C$ be $n$-input comparator networks. Suppose that $S_{A} \subseteq S_{B}$ and $B \oplus C$ is an $n$-input sorting network. Then there exists a comparator network $C'$ with the same depth as $C$ such that $A \oplus C'$ is an $n$-input sorting network.
\end{theorem}

\begin{proof}
See proof of the more general Theorem~\ref{MinimalPermutationOutputSetTheorem}.
\end{proof}

Knuth~\cite{Knuth73} has shown that comparator networks are just as powerful as generalized comparator networks. He shows that the group of generalized comparator networks is closed under permutation. Intuitively, we would like to strengthen the result of Theorem~\ref{MinimalOutputSetTheorem} by considering permutations of output sets. Before we present this result, we need the following lemma to prove it.

\begin{lemma}
\label{lemma:PermutationSet}
Let $\pi \in \Pi_n$, $x \in I_n$, and $C$ be a comparator network such that $V_C(\pi(x))$ is sorted. Then $\pi(V_{\pi^{-1}(C)}(x))$ is sorted, where $\pi^{-1}(C)$ is a generalized comparator network.
\end{lemma}

\begin{proof}
Let $x = \langle x_1, x_2, \dots, x_n \rangle$ and $P_C(\pi(x)) = \langle p_1, p_2, \dots, p_n \rangle$. Applying $\pi^{-1}$ to the equality yields that $P_{\pi^{-1}(C)}(\pi^{-1}(\pi(x))) = P_{\pi^{-1}(C)}(x) = \langle \pi^{-1}(p_{\pi^{-1}(1)}), \dots, \pi^{-1}(p_{\pi^{-1}(n)}) \rangle$. Applying $\pi$ to the equality yields $ \pi( P_{\pi^{-1}(C)}(x) ) = \langle p_1, p_2, \dots, p_n \rangle = P_C(\pi(x))$. From  Definition~\ref{def:ComparatorNetworkEvaluation} of the functions $V_C(x)$ and $P_C(x)$, we now have that $\pi( V_{\pi^{-1}(C)}(x) ) = V_C(\pi(x))$. From the hypothesis we know that $V_C(\pi(x))$ is sorted, hence we conclude that $\pi(V_{\pi^{-1}(C)}(x)) = V_C(\pi(x))$ is sorted.
\end{proof}

%\noindent
Theorem~\ref{MinimalOutputSetTheorem} tells us that if we can extend the comparator network $B$ to a sorting network by appending $l$ levels to it then we can extend any network $A$ such that $S_A \subseteq S_B$ by appending $l$ levels to it. We now extend this result by weakening the constraint $S_A \subseteq S_B$. We show that it is enough to find one permutation $\pi \in \Pi_n$ such that $\pi(S_A) \subseteq S_B$ to claim that if we can extend the comparator network $B$ to a sorting network by appending $l$ levels to it then we can extend any network $A$ by appending $l$ levels to it.

\begin{theorem}
\label{MinimalPermutationOutputSetTheorem}
Let $A$, $B$ and $C$ be $n$-input comparator networks, and $\pi \in \Pi_n$ such that $ \pi(S_{A}) \subseteq S_{B}$ and $B \oplus C$ is an $n$-input sorting network. Then there exists a comparator network $C'$ with the same depth as $C$ such that $A \oplus C'$ is an $n$-input sorting network.
\end{theorem}

\begin{proof}
From the hypothesis we know that there exists $C$ such that $B \oplus C$ is an $n$-input sorting network. From $\pi (S_{A}) \subseteq S_{B}$ we deduce that $V_C(\pi(x))$ is sorted for all $x \in S_{A}$ because $\pi(x) \in S_{B}$. Applying Lemma~\ref{lemma:PermutationSet} to all $x \in S_{A}$ and $C$ we deduce that $\pi(V_{\pi^{-1}(C)}(x))$ is sorted. Hence $A \oplus \pi^{-1}(C)$ is a generalized $n$-input sorting network of depth $k$. Finally we apply Theorem~\ref{FloydKnuthGeneralizedTheorem} to the generalized sorting network $A \oplus \pi^{-1}(C)$ to show that there exists a comparator network $C'$ with the same depth as $C$ such that $A \oplus C'$ is an $n$-input sorting network.
\end{proof}

\begin{definition}
\label{def:min_pi}
Let $X$ be a set of output sets of $n$-input comparator networks. Define the set of all minimal representative output sets up to permutation of $X$ as $MinPi ( X ) = \{ S_A ~\lvert~ S_A \in X ~:~\nexists~ S_B \in X, \pi \in \Pi_n : B < A, \pi(S_B) \subseteq S_A  \}$, where by $B < A$ we denote the lexicographic order of networks, as described by Parberry~\cite{Parberry89}. Let the set of all output sets of $n$-input comparator networks of depth $d$ be defined as $G_{n,d}$. Let the set of all minimal representative output sets of $n$-input comparator networks of depth $d$ up to permutation be defined as $S_{n,d} = MinPi(G_{n,d})$.
\end{definition}

\begin{definition}
\label{def:filters}
The set $X_n$ of $n$-input comparator networks is a complete set of filters iff for any $n$-input sorting network of depth $d$ there exists one of the form $C : C'$ of depth $d$ for some $C \in X_n$. We would also denote the set of all complete sets of filters of that contain only $n$-input comparator networks with exactly $i$ levels as $F_{n, i} = \{ X_n ~|~ X_n ~is~a~complete~set~of~filters~and~ C \in X_n \implies |C| = i \}$ .
\end{definition}

\subsection{Known Non-Overlapping Theory}
\label{sec:background:theory:non_overlap}

\begin{definition}
\label{def:network:reflect}
Let $ x = \langle x_1, x_2, \dots, x_n \rangle \in I_n $ then $ \overline{x^R} = \langle \overline{x_n}, \overline{x_{n-1}}, \dots, \overline{x_1} \rangle $ where $x_i \in \{0, 1\}$ and $\overline{0} = 1$ and $\overline{1} = 0$.
Let $ L $ be a level the its reflection $L^R = \{ \langle n - j + 1, n - i + 1 \rangle  ~|~ \langle i, j \rangle \}$.
Let $ C = \langle L_1, L_2, \dots, L_d, n \rangle $ be a comparator network then its reflection $C^R = \langle L_1^R, L_2^R, \dots, L_d^R, n \rangle $.
\end{definition}

\begin{lemma}
\label{lemma:network:reflect}
Let $ C $ be a comparator network then $x \in S_C \iff \overline{x^R} \in S_{C^R} $.
\end{lemma}

\begin{proof}
Refer to the proof of Lemma~8 in \cite{BundalaCCSZ14_Optimal_Depth} by Michael Codish.
\end{proof}

\begin{lemma}
\label{lemma:reflect}
Let $R_{n, i}$ be the set of minimal representative up to permutation and reflection itemsets within $G_{n}$. Then $R_{n, i} \in F_{n, i}$.
\end{lemma}

\begin{proof}
Refer to section~4.2 in \cite{BundalaCCSZ14_Optimal_Depth}.
\end{proof}

\begin{definition}
\label{def:omega}
The set of inputs $B \upharpoonright w = \{ \langle x_1, \dots, x_n \rangle \in I_n ~|~ x_1 = \dots x_l = 0, x_{n - r + 1} = \dots = x_{n} = 1, l + r = w \} $.
\end{definition}

\begin{remark}

Bundala~\cite{BundalaCCSZ14_Optimal_Depth} noted that if there exists an input set $B \in I_n$ for which no $n$-input sorting network of depth $d$ exists then there does not exist an $n$-input sorting network. The input sets considered by his method are of the form $B \upharpoonright w $,

Note, that $w$ is strongly dependant on $n$.
\end{remark}

%\clearpage
\section{Algorithm Modification}

In this section we present a modification to the algorithm \cite{Marinov:ExtremalSets:Permutation} for finding minimal itemsets up to permutation that allows us to find the minimal up to permutation and reflection as per Definition~\ref{def:network:reflect} and Lemma~\ref{lemma:reflect}. The pseudo code of the modified version is presented in Algorithm~\ref{algo:min_pi_refl}.

\subsection{Detailed Description}

The first thing the algorithm does is given a dataset $F$ it makes sure that for every itemset $F_i \in F$ its reflection $F_i^R$ is also in $F$, if this is not the case then we add it to the input dataset $F$. We assume, that the version of the algorithm for finding minimal sets up to permutation returns an array of $|F|$ integers $ subset\_of $ such that  if the itemset $F_i$ is minimal up to permutation we have $subset\_of[i] = i$; otherwise there exists a permutation $\pi \in \Pi_n$ such that $ \pi(F[ subset\_of[i] ]) \subseteq F[i] $. This array gives us detailed information about which itemset is a subset (up to permutation) of another. We use this extra information to remove the ones which are non-minimal up to reflection from the ones which are minimal only up to permutation. We do this by iterating through the list of itemsets that are minimal up to permutation and for each itemset at index $i$ we find the index $reflect[i]$ of its reflected itemset within $F$. Then we traverse the $subset\_of[ reflect[i] ]$ until we reach the index of an itemset that is minimal up to permutation. If the index at which we arrived is smaller than the index  $i$ then we mark the minimal up to permutation itemset $F_i$ as non-minimal up to reflection because we choose the lexicographically (index-wise) smallest itemset as representative up to reflection. This modified algorithm is easily proven to return the minimal up to permutation and reflection itemsets within $F$ by using Definition~\ref{def:network:reflect} and Lemma~\ref{lemma:reflect}.

\subsection{Complexity Analysis}
The worst case time and space complexity of this modified algorithm are the same as the unmodified version because all we do is add an extra $\BigO{|F|}$ time and space to the existing approach. Hence the worst time and space complexity are $\BigO{ \frac{r \times n! \times ||F||}{P} }$ (using $P$ parallel threads) and $ \BigO{ ||F|| + r \times n^2 }$ respectively.

\section{Experimental Evaluation}
\label{sec:experiments}

\begin{figure}
\centering
\resizebox{\linewidth}{!}{%
\begin{tabular}{c|c||c|c|c|c|c|c|c|c|c|c|c|}

$ \mathbf{ n } $					& $5$ & $6$ & $7$ & $8$ & $ 9 $ & $ 10 $ & $ 11 $ & $ 12 $ & $ 13 $ & $ 14 $ & $ 15 $ & $ 16 $ \\
$ \mathbf{ d } $					& $4$ & $4$ & $5$ & $5$ & $ 6 $ & $ 6 $ & $ 7 $ & $ 7 $ & $ 8 $ & $ 8 $ & $ 8 $ & $ 8 $ \\
$ \mathbf{ \omega } $				& $ 2 $ & $ 2 $ & $ 2 $ & $ 3 $ & $ 3 $ & $ 4 $ & $ 4 $ & $ 5 $ & $ 3 $ & $ 4 $ & $ 7 $ & $ 7 $  \\

\hline

$ \mathbf{ |G_{n}| } $		& $ 26 $ & $ 76 $ & $ 232 $ & $ 764 $ & $ 2\,620 $ & $ 9\,496 $ & $ 35\,696 $ & $ 140\,152 $ & $ 568\,504 $ & $ 2\,390\,480 $ & $ 10\,349\,536 $ & $ 46\,206\,736 $  \\

\hline

$ \mathbf{ |R_{n, 1}| } $			& $ 1 $ & $ 1 $ & $ 1 $ & $ 1 $ & $ 1 $ & $ 1 $ & $ 1 $ & $ 1 $ & $ 1 $ & $ 1 $ & $ 1 $ & $ 1 $  \\

\hline

$ \mathbf{ |R_{n, 2}| } $			& $ 4 $ & $ 5 $ & $ 8 $ & $ 12 $ & $ 22 $ & $ 21 $ & $ 48 $ & $ 50 $ & $ 117 $  & $ 94 $ & $ 262 $ & $ 211 $  \\
$ \mathbf{ |R_{n, 2} \upharpoonright {\omega}| } $	& $ 3 $ & $ 2 $ & $ 3 $ & $ 6 $ & $ 13 $ & $ 12 $ & $ 20 $ & $ 24 $ & $ 103 $  & $ 66 $ & $ 83 $ & $ 200 $  \\

$ \mathbf{ \floor{ \frac{ |R_{n, 2}| }{ |R_{n, 2} \upharpoonright {\omega}| }   }  } $	& $ 1.33 $ & $ 2.5 $ & $ 2.67 $ & $ 2.00 $ & $ 1.69 $ & $ 1.75 $ & $ 2.4 $ & $ 2.08 $ & $ 1.14 $ & $ 1.42 $ & $ 3.16 $ & $ 1.06 $  \\

\hline

%1\,733\,851
$ \mathbf{ |R_{n, 3}| } $			& $ 4 $ & $ 4 $ & $ 52 $ & $ 38 $ & $ 1\,554 $ & $ 3\,169 $ & $ 55\,722 $ & $ 117\,517 $ & $  $ & $  $ & $  $ & $  $  \\
$ \mathbf{ |R_{n, 3}^{\omega}| } $	& $ 4 $ & $ 4 $ & $ 27 $ & $ 55 $ & $ 685 $ & $ 971 $ & $ 12\,025 $ & $ 38\,758 $ & $ 2\,403\,835 $  & $  $ & $  $ & $  $  \\

%38.36
$ \mathbf{ \floor{ \frac{ |R_{n, 2}| * |G_{n}| }{ |R_{n, 3}| }   }  } $	& $ 26 $ & $ 95 $ & $ 35.69 $ & $ 241.26 $ & $ 37.09 $ & $ 62.93 $ & $ 30.75 $ & $ 59.63 $ & $  $ & $  $ & $  $ & $  $  \\ 

$ \mathbf{ \floor{ \frac{ |R_{n, 2}| * |G_{n}| }{ |R_{n, 3}^{\omega} | }   }  } $	& $ 26 $ & $ 95 $ & $ 68.74 $ & $ 166.69 $ & $ 84.15 $ & $ 205.37 $ & $ 142.49 $ & $ 180.80 $ & $ 27.67 $  & $  $ & $  $ & $  $  \\

\hline 
\end{tabular}
}
\caption{ Experimental evaluation summary presenting the sizes of the number of networks whose outputs are minimal and representative up to permutation and reflection of depths one, two and three for $n \leq 17$. The rows $ \mathbf{ \floor{ \frac{ |R_{n, 2}| * |G_{n}| }{ |R_{n, 3}| }   }  } $	 and $ \mathbf{ \floor{ \frac{ |R_{n, 2}| * |G_{n}| }{ |R_{n, 3}^{\omega} | }   }  } $ demonstrates the expected speedup of the existing \cite{BundalaCCSZ14_Optimal_Depth} algorithm for finding sorting networks of optimal depth by fixing the first three layers, rather than only the first two. The row $ \mathbf{ |R_{n, 2} \upharpoonright {\omega}| } $ shows the necessary second layer networks that Bundala's approach needs to consider when proving that no $n$-input sorting network of depth $d$ exists. }
\label{fig:experiments}
\end{figure}

We have summarized the results of our experiments in Figure~\ref{fig:experiments}. We show the sizes of the sets $R_{n, 1}$, $R_{n, 2}$ for all $n < 17$ and $R_{n, 3}$ for all $n < 14$. The sizes of $R_{n, 1}$ and $R_{n, 2}$ match exactly to the ones presented by Bundala~\cite{BundalaCCSZ14_Optimal_Depth} whereas the set $R_{n, 3}$ is novel to the work presented in this paper. The row $\mathbf{ \floor{ \frac{ |R_{n, 2}| * |G_{n}| }{ |R_{n, 3}| } } }$ presents the expected speedup of Bundala's algorithm if the first three layers are to be fixed rather than only the first two --- as is described in \cite{BundalaCCSZ14_Optimal_Depth}. This is technique is most useful when using Bundala's method to find the satisfiable instances when checking if an $n$-input sorting network of depth $d + 1$ exists; i.e. for $n = 12$ we do expect Bundala's program to execute about $60$ times faster to determine that there exists a sorting network of depth $d + 1 = 8$.

Bundala~\cite{BundalaCCSZ14_Optimal_Depth} noted that if there exists an set of inputs $B \in I_n$ for which no $n$-input sorting network of depth $d$ exists then there does not exist an $n$-input sorting network. The input sets considered by his method are of the form $B \upharpoonright w $, recall from Definition~\ref{def:omega}. In Figure~\ref{fig:experiments} we present the row $ \mathbf{ |R_{n, 2} \upharpoonright {\omega}| } $ which presents the minimal up to permutation and reflection outputs of depth two that are restricted to be of that certain form. We use these when proving that an $n$-input sorting network of depth $d$ does not exist; i.e. Bundala's approach needs to consider only $ \mathbf{ |R_{n, 2} \upharpoonright {\omega}| } $ set of fixed two-layered networks. Moreover, when generating the third layer $R_{n,3}^w$ is achieved by applying all levels to the itemsets from $ \mathbf{ |R_{n, 2} \upharpoonright {\omega}| } $. Hence, adapting the Bundala's algorithm to fix the first three layers, we would expect a speedup factor of $ \mathbf{ \floor{ \frac{ |R_{n, 2}| * |G_{n}| }{ |R_{n, 3}^{\omega} | }   }  } $ for the unsatisfied instances; i.e. for $n = 12$ we expect the modified version of his algorithm to find that no sorting network of depth seven exists about $180$ times faster by fixing the first three layers with $w = 5$.

\subsection{Environment Setup}
In all of the conducted experiments we used a computer with four Intel Xeon CPU E7- 4820 processors. Each CPU has 8 cores clocked at 2.00GHz, equipped with 8MB of third level cache and 128GB of main memory. Note that our experiments investigate the case when the entire data structure fits in main memory.

\subsection{Implementation Verification}
The correctness of our program was verified by calculating $R_{n, 2}$ and comparing to existing results \cite{CodishCS14:Two_Layer_Prefix} \cite{BundalaCCSZ14_Optimal_Depth}  (referred to as $R_n$). We have verified this for all $n < 17$.

\section{Conclusion and Future Work}

%\noindent
This paper has extended the knowledge of the structure of comparator networks when constrained to the problem of finding minimal depth sorting networks. The current state of the art algorithm for finding optimal depth sorting networks fixes the first two layers, formulates the problem as a SAT encoding and then uses an existing SAT solver to find the answer. Using the work presented in this paper, we can fix the first three layers of a comparator network and then construct the SAT encoding. In the presented experiments we managed to find the three layer networks for all $n \leq 13$, where for $n = 12$ we do expect Bundala's algorithm to execute around $180$ times faster when the first three layers are fixed in comparison to when only the first two are fixed.

For future work, we would like to improve the memory usage of the algorithm for finding minimal itemsets up to permutation and reflection, as currently it requires the whole dataset to fit into main memory. This is the primary reason why we do not present results for any $n \geq 17$, as we had access to a machine with only $128GB$ of main memory.

%\clearpage
\section{Acknowledgements}

%\noindent
%Almost all calculations were performed on the Lonsdale cluster maintained by the Trinity Centre for High Performance Computing. This cluster was funded through grants from Science Foundation Ireland.

%\noindent
Work supported by the Irish Research Council (IRC).

\bibliographystyle{elsarticle-num}
\bibliography{Third_Level}

%\clearpage

\begin{algorithm} [H]
\SetAlgoNoLine

\KwIn{ Dataset $F = \{ F_0, F_1, \dots, F_{r-1} \}$ over the domain $D = \{d_1, d_2,\dots, d_n\}$ and the degree of parallelism $P$ }

\KwOut{ The minimal itemsets within the dataset $F$ up to permutation of $D$. i.e. $Min_{\pi}(F)$ }

{
	\tcc{ We add the reflections of the itemsets that are missing from the dataset $F$. }
	\nl $F \longleftarrow F \bigcup F^R$\;

	\tcc{ Remember the reflect indexes. }
	\For{ $ i \longleftarrow 0$ \KwTo $ |F| $ }
	{
		\nl $reflect[i] \longleftarrow i'$ such that $F[i]^R = F[i']$\;
	}

	\nl $ subset\_of \longleftarrow Find-Min-Rep-Perm(F, P)$\;

	\nl \For{ $ i \longleftarrow 0$ \KwTo $ |F| $ }
	{
		\nl $is\_min\_pi[i] \longleftarrow false$\;
		
		\tcc{ Check if $F[i]$ is minimal over $F$ up to permutation. }
		\nl \If{$subset\_of[i] = i $}
		{
			\tcc{We set the itemset $F_i$ as minimal up to permutation.}
			\nl $is\_min\_pi[i] \longleftarrow true$\;
			
			\tcc{ We start with the reflection of $F[i]$ and work our way following the $subset\_of$ path to a minimal up to permutation itemset. }
			$item \longleftarrow reflect[i] $
			\nl \While{$ \subset\_of[item] \neq item $}
			{
				\nl $item \longleftarrow subset\_of[item]$\;
			}

			\nl $is\_min\_refl[i] \longleftarrow true $\;
			
			\nl \If{$ item < i $}
			{
				\tcc{We set the itemset $F_i$ as non-minimal up to reflection, because we choose the lexicographically smallest 
				(i.e the one with a least index in the dataset) to be representative up to reflection.}
				\nl $is\_min\_refl[i] \longleftarrow false $\;
			}
		}
	}
	\nl \KwRet $ \{ F_i \in F$ $\lvert$ $is\_min\_pi[i] \And is\_min\_refl[i] \}$\;
}
\caption{Pseudo code for finding the minimal up to permutation itemsets $M$ of the input dataset $F = \{  F_0, F_1, \dots, F_{r-1} \}$ using $T$ threads, where every $F_i \in F$ is an itemset over the domain $D = \{  d_1, d_2, \dots, d_n \}$. We present a subroutine Find-Min-Rep-Perm which identifies the minimal representative itemsets of $F$ using $T$ parallel threads. It is important to note that in the Thread-Functor subroutine the variables $index$ and $is\_min$ are passed to the  by reference, meaning that they are shared between threads.}
\label{algo:min_pi_refl}
\end{algorithm}

\end{document}